\documentclass[sigplan,9pt,acmthm=false]{acmart}
\usepackage[english]{babel}
\usepackage[T1]{fontenc}
\usepackage[utf8]{inputenc}
\usepackage{mathtools}
\usepackage[linesnumbered,vlined]{algorithm2e}
\usepackage{tikz}
\usetikzlibrary{arrows.meta}
\usetikzlibrary{fit}
\usetikzlibrary{shapes.geometric}
\usetikzlibrary{decorations.pathreplacing}
\usetikzlibrary{decorations.pathmorphing}
\usetikzlibrary{calc}
\usetikzlibrary{automata}

\newcommand{\PP}{\mathsf{P}}
\newcommand{\R}{\mathbb{R}}

\newcommand{\N}{\mathbb{N}}
\newcommand{\Cx}{\mathbb{C}}
\newcommand{\Q}{\mathbb{Q}}
\newcommand{\Qbar}{\mathbb{A}}

\newcommand{\K}{\mathbb{K}}

\newcommand{\set}[1]{\left\{#1\right\}}
\newcommand{\IdComp}[1]{{#1}_0}
\DeclareMathOperator{\im}{im}
\DeclareMathOperator{\Gr}{Gr}

\DeclareMathOperator{\rank}{rk}

\DeclareMathOperator{\GL}{GL}

\DeclareMathOperator{\diag}{diag}
\DeclareMathOperator{\Span}{span}

\newcommand{\Zcl}[2][]{\overline{#2}^{#1}}

\newcommand{\Gen}[1]{\left\langle#1\right\rangle}
\newcommand{\ExtAlg}[2][]{\Lambda^{#1}{#2}}

\allowdisplaybreaks

\newcommand{\checkthat}[1]{\textcolor{red}{[check that]#1}}
\newcommand{\ie}{\emph{i.e.}}

\theoremstyle{acmplain}
\newtheorem{theorem}{Theorem}
\newtheorem{proposition}[theorem]{Proposition}
\newtheorem{corollary}[theorem]{Corollary}
\newtheorem{lemma}[theorem]{Lemma}
\theoremstyle{acmdefinition}

\SetAlgoNoEnd
\DontPrintSemicolon
\SetKw{KwFrom}{from}
\SetKw{KwAnd}{and}
\SetKwInOut{Input}{input}
\SetKwInOut{Output}{output}

\begin{document}

\copyrightyear{2018} 
\acmYear{2018} 
\setcopyright{acmcopyright}
\acmConference[LICS '18]{LICS '18: 33rd Annual ACM/IEEE Symposium on Logic in Computer Science}{July 9--12, 2018}{Oxford, United Kingdom}
\acmBooktitle{LICS '18: LICS '18: 33rd Annual ACM/IEEE Symposium on Logic in Computer Science, July 9--12, 2018, Oxford, United Kingdom}
\acmPrice{15.00}
\acmDOI{10.1145/3209108.3209142}
\acmISBN{978-1-4503-5583-4/18/07}

\title{Polynomial Invariants for Affine Programs}
\author{Ehud Hrushovski}
\affiliation{
  \department{Mathematical Institute}
  \institution{Oxford University, UK}
}
\email{ehud.hrushovski@maths.ox.ac.uk}
\author{Jo\"el Ouaknine}
\affiliation{
  \institution{Max Planck Institute for Software Systems}
  \city{Saarland Informatics Campus}
  \country{Germany}
}
\affiliation{
  \department{Department of Computer Science}
  \institution{Oxford University, UK}
}
\email{joel@mpi-sws.org}
\author{Amaury Pouly}
\affiliation{
  \institution{Max Planck Institute for Software Systems}
  \city{Saarland Informatics Campus}
  \country{Germany}
}
\email{pamaury@mpi-sws.org}
\author{James Worrell}
\affiliation{
  \department{Department of Computer Science}
  \institution{Oxford University, UK}
}
\email{jbw@cs.ox.ac.uk}

\begin{abstract}
  We exhibit an algorithm to compute the strongest polynomial (or
  algebraic) invariants that hold at each location of a given affine
  program (i.e., a program having only non-deterministic (as opposed
  to conditional) branching and all of whose assignments are given by
  affine expressions). Our main tool is an algebraic result of
  independent interest: given a finite set of rational square matrices
  of the same dimension, we show how to compute the Zariski closure of
  the semigroup that they generate.
\end{abstract}

\maketitle

\begin{acks}
Jo\"el Ouaknine was supported by ERC grant AVS-ISS
(648701), and James Worrell was supported by EPSRC Fellowship EP/N008197/1.
\end{acks}

\section{Introduction}

\paragraph{Invariants} 
Invariants are one of the most fundamental and useful notions in
the quantitative sciences, appearing in a wide range of contexts, from gauge theory,
dynamical systems, and control theory in physics, mathematics, and
engineering to program verification, static analysis, abstract
interpretation, and programming language semantics (among others) in
computer science. In spite of decades of scientific work and
progress, automated invariant synthesis remains a topic of active
research, particularly in the fields of theorem proving and program
analysis, and plays a central role in methods and tools seeking to establish
correctness properties of computer programs; 
see, e.g.,~\cite{KCBR18}, and particularly Sec.~8 therein.

\paragraph{Affine Programs}
Affine programs are a simple kind of nondeterministic imperative
programs (which may contain arbitrarily nested loops) in which the
only instructions are assignments whose right-hand sides are affine
expressions, such as $x_3:=x_1-3x_2+7$. A conventional imperative
program can be abstracted to an affine program by replacing
conditionals with nondeterminism and conservatively over-approximating
non-affine assignments, see, e.g., \cite{BradleyM07}. In doing so,
affine programs enable one to reason about more complex programs; a
particularly striking example is the application of affine programs to
several problems in inter-procedural
analysis~\cite{BradleyM07,Muller-OlmS04,Muller-OlmS04b,GulwaniN03}.

\paragraph{Affine Invariants}
An affine invariant for an affine program with $n$ variables assigns
to each program location an affine subspace of $\R^n$ such that the
resulting family of subspaces is preserved under the transition
relation of the program.  Such an invariant is specified by giving a
finite set of affine relations at each location.  The strongest (i.e.,
smallest with respect to set inclusion) affine invariant is obtained
by taking the affine hull of the set of reachable configurations
(i.e., values of the program variables) at each program location.
Equivalently, the strongest affine invariant is determined by giving,
for each program location, the set of all affine relations holding at
that location.

An algorithm due to Michael Karr in 1976~\cite{Karr76} computes the
strongest affine invariant of an affine program.  
A more efficient reformulation of Karr's
algorithm was given by M\"{u}ller-Olm and Seidl~\cite{Muller-OlmS04},
who moreover showed that if the class of affine programs is augmented
with equality guards then it becomes undecidable whether or not a
given affine relation holds at a particular program location.  A
randomised algorithm for discovering affine relations was proposed by
Gulwani and Necula~\cite{GulwaniN03}.

\paragraph{Polynomial Invariants}
A natural and more expressive generalisation of affine invariants are
polynomial invariants.  A polynomial invariant assigns to each program
location a variety (or algebraic set, i.e., positive Boolean
combination of polynomial equalities) such that the resulting family
is preserved under the transition relation of the program.  A
polynomial invariant is specified by giving a set of polynomial
relations that hold at each program location.  The strongest
polynomial invariant (i.e., smallest variety with respect to set
inclusion) is obtained by taking the Zariski
closure of the set of reachable configurations in each location.

The problem of computing polynomial invariants for affine programs and
related formalisms has been extensively studied over the past fifteen
years years; see,
e.g.,~\cite{CacheraJJK14,Colon07,SankaranarayananSM04,Kapur13,Rodriguez-CarbonellK07,Rodriguez-CarbonellK07b,
  OliveiraBP16,KovacsJ05,Kovacs08,HumenbergerJK18,KCBR18,OBP16}.
However, in contrast to the case of affine invariants, as of yet no
method is known to compute the \emph{strongest} polynomial invariant, i.e.,
(a basis for) the set of all polynomial relations holding at each
location of a given affine program.  Existing methods are either
heuristic in nature, or only known to be complete relative to
restricted classes of invariants or programs.  For example, it is
shown in~\cite{Muller-OlmS04} (see
also~\cite{Rodriguez-CarbonellK07b}) that Karr's algorithm can be
applied to compute the smallest polynomial invariant that is specified
by polynomial relations of a fixed degree $d$.  (The case of affine
invariants corresponds to $d=1$.)  The paper~\cite{OliveiraBP16} gives
a method that finds all polynomial invariants for a highly restricted
class of affine programs (in which all linear mappings have positive
rational eigenvalues).  The approach
of~\cite{Kovacs08,HumenbergerJK18} via so-called P-solvable loops does
not encompass the whole class of affine programs either (although it does
allow to handle certain classes of programs with polynomial
assignments)~\cite{Kovacspersonal}.

\paragraph{Main Contribution}
In this paper we give a method to compute the set of \emph{all} polynomial
relations that hold at a given location of an affine program, or in
other words the strongest polynomial invariant.  The
output of the algorithm gives for each program location a finite basis
of the ideal of all polynomial relations holding at that location.  

Our main tool is an algebraic result of independent interest: we give
an algorithm that, given a finite set of rational square matrices of
the same dimension, computes the Zariski closure of the semigroup that
they generate.  Our algorithm generalises (and uses as a subroutine)
an algorithm of Derksen, Jeandel, and Koiran~\cite{DerksenJK05} to
compute the Zariski closure of a finitely generated group of
invertible matrices.%
\footnote{Related to this, Corollary~3.7 and Lemma~3.6a in~\cite{HR} reduce the
question of computing the Zariski closure of a finitely generated
group of invertible matrices to that of finding multiplicative relations among
diagonal matrices. Note that if one begins with rational matrices,
then such relations can be found simply using prime decomposition of
the entries.}

Our procedure for computing the Zariski closure of a matrix semigroup
also generalises a result of Mandel and Simon~\cite{MandelS77} and,
independently, of Jacob~\cite{Jacob77,Jacob78}, to the effect that it is decidable
whether a finitely generated semigroup of rational matrices is finite.
Note that a variety that is given as the zero set of a polynomial
ideal $I\subseteq \K[x_1,\ldots,x_n]$ is finite just in case the
quotient $\K[x_1,\ldots,x_n]/I$ is finite-dimensional as a vector
space over $\K$~\cite[Chapter~5, Sec.~3]{CoxLS}).  The latter
condition can be checked by computing a Gr\"{o}bner basis for $I$.

As mentioned above, we make use of the result of~\cite{DerksenJK05}
that one can compute the Zariski closure of the group generated by a
finite set of invertible matrices.  That result itself relies on several
non-trivial mathematical ingredients, including results of
Masser~\cite{Mas88} on computing multiplicative relations among given
algebraic numbers and Schur's theorem that every finitely generated
periodic subgroup of the general linear group $\GL_n(\Cx)$ is finite.  

Given a set of matrices $A\subseteq M_n(\Cx)$, we leverage these
group-theoretic results to compute the Zariski closure $\Zcl{\Gen{A}}$
of the generated semigroup $\Gen{A}$.  To this end we use multilinear
algebra as well as structural properties of matrix semigroups to identify
finitely many subsemigroups of $\Zcl{\Gen{A}}$ that can be used to
generate the entire semigroup.  Pursuing this approach requires that we
first generalise the result of~\cite{DerksenJK05} to show that one can
compute the Zariski closure of the group generated by a constructible
(as opposed to finite) set of invertible matrices. 

It is worth pointing out that whether a particular configuration is
reachable at a certain program location of a given affine program is
in general an undecidable problem---this follows quite straightforwardly
from the undecidability of the membership problem for finitely
generated matrix semigroups, discussed shortly. It is therefore
somewhat remarkable that the Zariski closure (i.e., the smallest
algebraic superset) of the set of reachable configurations at any
particular location nevertheless turns out to be a computable object.

\paragraph{Matrix Semigroups and Automata}
Decision problems for matrix semigroups have also been studied for
many decades, independently of program analysis. One of the most
prominent such is the Membership Problem, i.e., whether a given matrix
belongs to a finitely generated semigroup of integer matrices. An
early and striking result on this topic is due to Markov, who showed
undecidability of the Membership Problem in dimension $6$ in
1947~\cite{Mar47}. Later Paterson~\cite{Pat70} improved this result to
show undecidability in dimension $3$, while decidability
in dimension $2$ remains open.
A breakthrough was achieved in 2017 by
Potapov and Semukhin, who showed decidability of membership for
semigroups generated by \emph{nonsingular} integer $2 \times 2$
matrices~\cite{PS17}. By contrast, the Membership Problem was shown to
be polynomial-time decidable in any dimension by Babai \emph{et al.}\
for \emph{commuting} matrices over algebraic
numbers~\cite{BBCIL96}. As aptly noted by Stillwell, ``noncommutative
semigroups are hard to understand''~\cite{Sti16}. Matrix semigroup
theory also plays a central role in the analysis of weighted automata
(such as probabilistic and quantum automata, see, e.g.,
\cite{BJKP05,DerksenJK05}).

\paragraph{Abstract Interpretation and Other Approaches}
Polynomial\linebreak invariants are stronger (i.e., more precise) than affine
invariants.  Various other types of domains have been considered in
the setting of abstract interpretation, e.g., intervals, octagonal
sets, and polyhedra (see, e.g.,~\cite{CC77,CH78,Min01} and references
in~\cite{BradleyM07}). The precision of such domains in general is
incomparable to that of polynomial invariants. 

The computation of semialgebraic and o-minimal invariants has also been considered in
the context of discrete-time linear dynamical systems and linear loops (which can be
viewed as highly restricted instances of affine programs); see,
e.g.,~\cite{FOOPW17,ACOW18}. 

\section{Two Illustrative Examples}

We now present two simple examples to illustrate some of the ideas and
concepts that are discussed in this paper. Some of the notation and
terminology that we use is only introduced in later sections; should this
impede understanding, we recommend that the reader return to these examples
after having read Sections~\ref{sec:misc} and \ref{sec:affine-programs}.

As a first motivating example, consider the following linear loop:
\begin{align*}
& x := 3; \\ 
& y := 2; \\
&  \mbox{\textbf{while}}\ 2y-x \geq -2\ \mbox{\textbf{do}} \\
& \quad \begin{pmatrix} x\\ y\end{pmatrix} := 
\begin{pmatrix} 10& -8 \\ 6 &  -4 \end{pmatrix} 
\begin{pmatrix} x\\ y\end{pmatrix} ;
\end{align*}

This loop never halts, although this fact is perhaps not immediately
obvious. Here we show how the techniques developed in this paper can help
establish non-termination. To this end, we first turn our code
into an affine program consisting of two locations, as follows:
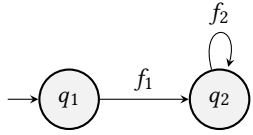
\begin{center}
\begin{tikzpicture}[->, >=stealth,node distance=2cm, every state/.style={thick, fill=gray!10},initial text=$ $]
\node[state,initial] (q1) {$q_1$};
\node[state, right of=q1] (q2) {$q_2$};
\draw
(q1) edge[above] node{$f_1$} (q2)
(q2) edge[loop above] node{$f_2$} (q2);
\end{tikzpicture}
\end{center}
Here $f_1$ is the constant affine function assigning $3$ to $x$ and $2$ to $y$,
whereas $f_2$ is the linear transformation associated with the matrix
appearing in our while loop. Note that we have discarded the loop
guard.

The collecting semantics of this affine program assigns to location $q_2$ the
set $S_{q_2} \subseteq \mathbb{Z}^2$ of all values taken by the pair
of variables $(x,y)$ in the unending execution of the program. As it
turns out, the real Zariski closure $\Zcl[\R]{S_{q_2}}$ of $S_{q_2}$ consists of the set 
\[ 
\{ (x,y) \in \mathbb{R}^2 : x - 9x^2 -y + 24xy -16y^2 = 0 \} \, . 
\]
By construction, this polynomial invariant is stable under $f_2$
and over-approximates the set $S_{q_2}$ of reachable $(x,y)$-configurations.
Verifying that all tuples in this variety moreover satisfy the guard $2y-x \geq -2$
is now a simple exercise in high-school algebra, from which one
concludes that our original loop will indeed never terminate.

For our second example, define $G := \Zcl[\R]{\Gen{S,T,E}}$ to be the
matrix semigroup obtained as the  
real Zariski closure of the semigroup generated by $S,T$ and
$E$, where
\[ S:=\begin{pmatrix} 0&-1\\1 & 0 \end{pmatrix} \, ,
\quad 
  T := \begin{pmatrix} 1&1\\0 & 1 \end{pmatrix} \, ,
\quad
  E := \begin{pmatrix} 1&0\\0 & 0 \end{pmatrix} \, .
\]
We show that
$G=\{ M \in M_2(\mathbb{R}) : \det(M)=1 \mbox{ or } \det(M)= 0
\}$ and in the process illustrate (in a very simple setting) the
approach of computing the Zariski closure of a matrix semigroup by
order of decreasing rank.  This approach underlies the algorithm
described in Sec.~\ref{sec:main-algorithm}.

Consider first $G':=\{ M \in G : \rank(M)=2 \}$.  From the fact that
the set of singular matrices in $M_2(\mathbb{R})$ is Zariski closed,
one can show that $G'=\{M \in \Zcl[\R]{\Gen{S,T}} : \rank(M)=2\}$.
Now it is well known that $S$ and $T$ generate the semigroup
$\mathrm{SL}_2(\mathbb{Z})$ of $2\times 2$ integer matrices of
determinant $1$ and that the real Zariski closure of
$\mathrm{SL}_2(\mathbb{Z})$ is the semigroup
$\mathrm{SL}_2(\mathbb{R})$ of $2\times 2$ \emph{real} matrices of
determinant $1$\footnote{The latter fact follows from the Borel
  density theorem~\cite[Sec. 4.5 and Sec. 7.0]{Morris01}, but can also be
  established directly by an elementary argument.}; hence
$G'= \mathrm{SL}_2(\mathbb{R})$.  More generally, we can use the
algorithm of Derksen, Jeandel, and Koiran~\cite{DerksenJK05} to
compute the Zariski closure of any finitely generated semigroup of
invertible matrices.

Now we consider the sub-semigroup $G''$ of singular matrices in $G$.
This is the real Zariski closure of the semigroup generated by the
(constructible) set of matrices
\[ \{ MEM', ME, EM : M,M' \in \mathrm{SL}_2(\mathbb{R}) \} \, . \] It is
straightforward to observe that this generating set already includes
all rank-$1$ matrices in $M_2(\mathbb{R})$ and hence that the
generated semigroup contains all singular matrices.  We conclude that
$G=G'\cup G''$ comprises all matrices in $M_2(\mathbb{R})$ of
determinant $0$ or $1$.
\section{Mathematical Background}
\label{sec:misc}
\subsection{Linear Algebra}
\hspace{\parindent}\textbf{Matrices.}  Let $\K$ be a field.  We denote
by $M_n(\K)$ the semigroup of square matrices of dimension $n$ with
entries in $\K$.  We write $\GL_n(\K)$ for the subgroup of $M_n(\K)$
comprising all invertible matrices.  Given a set of matrices
$A \subseteq M_n(\K)$, we denote by $\Gen{A}$ the sub-semigroup of
$M_n(\K)$ generated by $A$.  The rank of a matrix $A$ is denoted by
$\rank(A)$, its kernel by $\ker(A)$, and its image by $\im(A)$.

\textbf{Exterior Algebra and the Grassmannian.}  Given a vector\linebreak
 space
$V$ over the field $\K$, its exterior algebra $\ExtAlg{V}$ is a vector
space that embeds $V$ and is equipped with an associative, bilinear,
and anti-symmetric map
\[\wedge : \ExtAlg{V} \times \ExtAlg{V} \rightarrow \ExtAlg{V} \, . \]
We can construct $\ExtAlg{V}$ as a direct sum
\[ \ExtAlg{V} = \ExtAlg[0]{V} \oplus \ExtAlg[1]{V} \oplus
  \ExtAlg[2]{V} \cdots \, , \] where $\ExtAlg[r]{V}$ denotes the
$r^{\textit{th}}$-exterior power of $V$ for $r\in \mathbb{N}$, that
is, the subspace of $\ExtAlg{V}$ generated by $r$-fold wedge products
$v_1 \wedge \ldots \wedge v_r$ for $v_1,\ldots,v_r \in V$.  If $V$ is
finite dimensional, with basis $e_1,\ldots,e_n$, then a basis of
$\ExtAlg[r]{V}$ is given by $e_{i_1} \wedge\cdots \wedge e_{i_r}$,
$1\leq i_1 < \ldots < i_r \leq n$.  Thus $\ExtAlg[r]{V}$ has dimension
$\binom{n}{r}$ (where $\binom{n}{r}=0$ for $r>n$).

A basic property of the wedge product is that given vectors
$u_1,\ldots,u_r \in V$, $u_1 \wedge \ldots \wedge u_r \neq 0$ if and
only if $\{u_1,\ldots,u_r\}$ is a linearly independent set.
Furthermore given $w_1,\ldots,w_r \in V$ we have that
$u_1\wedge \ldots \wedge u_r$ and $w_1\wedge \ldots \wedge w_r$ are
scalar multiples of each other iff
$\Span(u_1,\ldots,u_r)= \Span(v_1,\ldots,v_r)$.  


The \emph{Grassmannian} $\Gr(r,n)$ is the set of $r$-dimensional
subspaces of $\Cx^n$.  By the above-stated properties of the wedge
product there is an injective function
\[\iota : \Gr(r,n) \rightarrow\ExtAlg[r](\Cx^n)\]
such that for any $W$, $\iota(W)=v_1\wedge\cdots\wedge v_r$ where
$v1,\ldots,v_r$ is an arbitrarily chosen basis of $W$. Note that given
two basis $v_1,\ldots,v_r$ and $u_1,\ldots,u_r$ of $W$, there exists
$\alpha\in\Cx$ such that
$v_1\wedge\cdots v_r=\alpha(u_1\wedge\cdots u_r)$. In other words, the
particular choice of a basis for $W$ only changes the value of
$\iota(W)$ up to constant. Given subspaces $W_1,W_2 \subseteq V$ we
moreover have $W_1\cap W_2 = 0$ iff
$\iota(W_1) \wedge \iota(W_2) \neq 0$.

\subsection{Algebraic Geometry}
\label{sub:alg-geo}
In this section we summarise some basic notions of algebraic geometry
that will be used in the rest of the paper.  

Let $\K$ be a field.  An \emph{affine variety} $X\subseteq\K^n$ is the
set of common zeros of a finite collection of polynomials, \ie{}, a
set of the form
\[X=\set{x\in\K^n:p_1(x)=p_2(x)=\cdots=p_\ell(x)=0}\, ,\] where
$p_1,\ldots,p_\ell \in \K[x_1,\ldots,x_n]$.  Given a polynomial ideal
$I\subseteq \K[x_1,\ldots,x_n]$, by Hilbert's basis theorem the set
\[\mathbf{V}(I) = \{x\in \K^n : \forall p \in I,\,p(x)=0 \}\] is a
variety, called the \emph{variety of $I$}.  The two main varieties of
interest to us are $X=M_n(\K)$, which we identify with affine space
$\K^{n^2}$ in the natural way, and $X=\GL_n(\K)$, which we identify
with the variety
\[ \{(A,y) \in \K^{n^2+1} : \det(A) \cdot y = 1\} \, . \]

Given an affine variety $X\subseteq\K^n$, the \emph{Zariski topology}
on $X$ has as closed sets the subvarieties of $X$, i.e., those sets
$A\subseteq X$ that are themselves affine varieties in $\K^n$.  For
example, $\{ a \in M_n(\K) : \rank(a) < r\}$ is a Zariski closed
subset of $M_n(\K)$, since for $a\in M_n(\K)$ we have $\rank(a)<r$ iff
all $r\times r$ minors of $a$ vanish.  Given an arbitrary set
$S\subseteq X$, we write $\Zcl{S}$ for its closure in the Zariski
topology on $X$.

It is straightforward that if $X \subseteq \Cx^n$ is a complex variety
then $X\cap \R^n$ is a real variety.  It follows that the Zariski
topology on $M_n(\R)$ coincides with the subspace topology induced on
$M_n(\R)$ by the Zariski topology $M_n(\Cx)$.  In particular, we can
compute the Zariski closure of a set of matrices $A \subseteq M_n(\R)$
by first computing the Zariski closure of $A$ in the complex variety
$M_n(\Cx)$ and then intersecting with $M_n(\R)$.

A set $S \subseteq X$ is \emph{irreducible} if
for all closed subsets $A_1,A_2\subseteq X$ such that
$S\subseteq A_1 \cup A_2$ we have either $S\subseteq A_1$ or
$S\subseteq A_2$.  It is well known that the Zariski topology on a
variety is Noetherian.  In particular, any closed subset $A$ of $X$
can be written as a finite union of \emph{irreducible components},
where an irreducible component of $A$ is a maximal irreducible closed
subset of $A$.

The class of \emph{constructible} subsets of a variety $X$ is obtained
by taking all Boolean combinations (including complementation) of
Zariski closed subsets.  Suppose that the underlying field $\K$ is
algebraically closed.  Since the first-order theory of algebraically
closed fields admits quantifier elimination, the constructible subsets
of $X$ are exactly the subsets of $X$ that are first-order definable
over $\K$.

Suppose that $X\subseteq\K^m$ and $Y\subseteq\K^n$ are affine
varieties.  A function $\varphi : X \rightarrow Y$ is called a
\emph{regular map} if it arises as the restriction of a polynomial map
$\K^m\rightarrow \K^n$.  Chevalley's Theorem states that if $\K$
algebraically closed and $\varphi : X \rightarrow Y$ is a regular map then
the image $\varphi(A)$ of a constructible set $A\subseteq X$ under
$\varphi$ is a constructible subset of $Y$.  This result also follows
from the fact that the theory of algebraically closed fields admits
quantifier elimination.

A regular map of interest to us is matrix multiplication
$M_n(\K) \times M_n(\K) \rightarrow M_n(\K)$.  In particular, we have
that for constructible sets of matrices $A,B \subseteq M_n(\K)$ the
set of products \[ A \cdot B := \{ ab : a \in A, b \in B \}\] is again
constructible.  Notice also that matrix inversion is a regular map
$\GL_n(\K) \rightarrow \GL_n(\K)$.  Thus if $A\subseteq \GL_n(\K)$ is
a constructible set then so is $A^{-1} := \{ a^{-1} : a \in A\}$.
Finally, the projection $(A,y) \mapsto A$ yields an injective regular
map $\GL_n(\K) \rightarrow M_n(\K)$.  Via this map we can identify
$\GL_n(K)$ with a constructible subset of $M_n(\K)$.

On several occasions we will use the facts that regular maps are
continuous with respect to the Zariski topology and that the image of
an irreducible set under a regular map is again irreducible.  In
particular, we have:
\begin{lemma}\label{lem:prod_irred_irred}
  If $X,Y\subseteq \GL_n(\Cx)$ are irreducible closed sets then
  $\Zcl{X\cdot Y}$ is also irreducible.
\end{lemma}

\subsection{Algorithmic Manipulation of Constructible Sets} 
\label{algorithmic_manipulation}
In this subsection we briefly recall some algorithmic constructions on
constructible subsets of a variety.  Here, and in the rest of the
paper unless noted otherwise, we assume that the underlying field is
$\Cx$ and that all ideals are generated by polynomials with algebraic
coefficients.

\textbf{Representing Constructible Sets.}  Consider a variety
$X\subseteq \Cx^n$ and let $I \subseteq \Cx[x_1,\ldots,x_n]$ be the
ideal of polynomials that vanish on $X$.  We represent Zariski closed
subsets of $X$ as zero sets of ideals in the coordinate ring
$\Cx[X]=\Cx[x_1,\ldots,x_n]/I$ of $X$.  
The coordinate ring of $M_n(\Cx)$
is just $\Cx[x_{1,1},\ldots,x_{n,n}]$ while the coordinate ring of
$\GL_n(\Cx)$ is
\[ \Cx[x_{1,1},\ldots,x_{n,n},y] / (\mathrm{det}(x_{i,j})y-1) \, . \]

Unions and intersections of Zariski closed subsets of $X$ respectively
correspond to products and sums of the corresponding ideals in
$\Cx[X]$.  
We furthermore represent constructible
subsets of $X$ as Boolean expressions over Zariski closed subsets.

\textbf{Irreducible Components.}  Let $A \subseteq X$ denote a Zariski
closed set that is given as the variety of an $I\subseteq\Cx[X]$.  If
$I = P_1 \cap \cdots \cap P_m$ is an irredundant decomposition of $I$
into primary ideals, then
$A=\mathbf{V}(P_1)\cup \ldots \cup \mathbf{V}(P_m)$ is a decomposition
of $A$ into irreducible components.  One can compute the primary
decomposition of an ideal using Gr\"{o}bner basis
techniques~\cite[Chapter 8]{BeckerW93}.

\textbf{Zariski Closure.}  At several points in our development, we will need to
compute the Zariski closure of a constructible subset of a variety.
Now an arbitrary constructible subset of a variety $X$ can be written
as a union of differences of closed subsets of $X$.  Thus it suffices
to be able to compute the closure of $A\setminus B$ for closed sets
$A,B\subseteq X$.  Furthermore, by first computing a decomposition of
$A$ as a union of irreducible closed sets, we may also assume that $A$
is irreducible.  But $A\subseteq \Zcl{A \setminus B}\cup (A\cap B)$;
thus by irreducibility of $A$ we have $\Zcl{A\setminus B} = \emptyset$
if $A\subseteq B$ and otherwise $\Zcl{A\setminus B} = A$.  An
algorithm (when using the representation above)
for computing the Zariski closure of a
constructible set, essentially following this recipe, is given
in~\cite[Theorem 1]{Koiran00}.

\textbf{Images under Regular Maps.}  One can use an algorithm for
quantifier elimination for the theory of algebraically closed fields
in order to compute the image of a constructible set under a regular
map.  An explicit algorithm for this task, using Gr\"{o}bner bases, is
given in~\cite[Sec.~4]{Schau07}.  

\textbf{Real Zariski Closure.} Given a complex variety
$V\subseteq\Cx^n$, the intersection $V\cap\R^n$, which is a real
variety, can be computed effectively. Indeed if $V$ is represented by
the ideal $I$ then $V\cap\R^n$ is represented by the ideal generated
by $\set{p_\R,p_{i\R}:p\in I}$ where $p_\R$ and $p_{i\R}$ respectively
denote the real and imaginary parts of the polynomial $p$. It is also
straightforward to verify that for a set $S \subseteq \R^n$, we have
$\Zcl[\R]{S} = \Zcl[\Cx]{S} \cap \R^n$.

\section{Algebraic Invariants for Affine Programs}
\label{sec:affine-programs}
An affine function $f:\Q^n\rightarrow \Q^n$ is a function of the form
$f(\boldsymbol{x})=A\boldsymbol{x}+\boldsymbol{b}$, where
$A \in M_n(\Q)$ and $\boldsymbol{b}\in\Q^n$.  We write
$\mathrm{Aff}_n(\Q)$ for the set of affine functions on $\Q^n$.

An \emph{affine program} of dimension $n$ is a tuple
$\mathcal{A}=(Q,E,q_{\mathrm{init}})$, where $Q$ is a finite set of
\emph{program locations},
$E \subseteq Q \times \mathrm{Aff}_n(\Q) \times Q$ is a finite set of
\emph{edges}, and $q_{\mathrm{init}}\in Q$ is the \emph{initial
  location}.

The \emph{collecting semantics} of an affine program $\mathcal{A}$
assigns to each location $q$ the set $S_q\subseteq\Q^n$ of all those
vectors that occur at location~$q$ in some execution of the program.
The family $\{S_q:q\in Q\}$ can be characterised as the least solution
of the following system of inclusions (see~\cite{Muller-OlmS04}):
\begin{eqnarray*}
X_{q_{\mathrm{init}}} & \supseteq & \{\boldsymbol{0}\} \\
X_q & \supseteq & f(X_p) \quad \mbox{for all $(p,f,q) \in E$} \, .
\end{eqnarray*}

Given $P\in \R[x_1,\ldots,x_n]$ we say that the polynomial relation
$P=0$ holds at a program location $q$ if $P$ vanishes on $S_q$.  We
are interested in the problem of computing for each location $q\in Q$
a finite set of polynomials that generate the ideal $I_q :=
\mathbf{I}(S_q) \subseteq \R[x_1,\ldots,x_n]$ of all polynomial
relations holding at location~$q$.  The variety corresponding to ideal
$I_q$ is $V_q:=\mathbf{V}(I_q)=\Zcl[\R]{S_q} \subseteq \R^n$, i.e.,
$V_q$ is the Zariski closure of $S_q$ regarded as a subset of real
affine space.  Thus the problem of computing the family of ideals
$I_q$ is equivalent to the problem of computing the (family of ideals
representing the) Zariski closure of the collecting semantics.

The indexed collection of varieties $\{ V_q : q \in Q \}$ defines
an \emph{invariant} in that for every edge $(p,f,q) \in E$ we have
$f(V_p) \subseteq V_q$.  This follows from the facts that
$f(S_p)\subseteq S_q$ and that $f$ is Zariski continuous.  By
construction we have that $\{ V_q : q \in V \}$ is the smallest
algebraic invariant of the program $\mathcal{A}$ such that
$\boldsymbol{0} \in V_{q_{\mathrm{init}}}$.

In the remainder of this section we reduce the problem of computing
the Zariski closure of the collecting semantics of an affine program
to that of computing the Zariski closure of a related semigroup of
matrices.  The idea of this reduction is first to replace each affine
assignment by a corresponding linear assignment by adding an extra
dimension to the program.  One then simulates a general affine program
by a program with a single location.

Consider an affine program $\mathcal{A}=(Q,E,q_{\mathrm{init}})$,
where the set of locations is $Q=\{q_1,\ldots,q_m\}$ and
$q_{\mathrm{init}}=q_1$.  For each edge $e=(q_j,f,q_i)$ we define a
square matrix $M^{(e)} \in M_{m(n+1)}(\Q)$ comprising an $m\times m$
array of blocks, with each block a matrix in $M_{n+1}(\Q)$.  If the
affine map $f$ is given by
$f(\boldsymbol{x})=A\boldsymbol{x}+\boldsymbol{b}$ then the $(i,j)$-th
block of $M^{(e)}$ is
\[ \begin{pmatrix} A & \boldsymbol{b} \\ 0 & 1 \end{pmatrix} \, , \] while all
other blocks are zero.  Notice that for $\boldsymbol{x} \in \Q^n$
we have
\begin{gather} \begin{pmatrix} A & \boldsymbol{b} \\ 0 & 1 \end{pmatrix} \begin{pmatrix}
  \boldsymbol{x}\\ 1 \end{pmatrix} = \begin{pmatrix}A\boldsymbol{x}+\boldsymbol{b}
  \\ 1 \end{pmatrix} = \begin{pmatrix} f(\boldsymbol{x}) \\ 1 \end{pmatrix}\, .
\label{eq:affine}
\end{gather}

Given $i\in\{1,\ldots,m\}$, define the projection
$\Pi_i : \Cx^{m(n+1)}\rightarrow \Cx^{n+1}$ by
$\Pi_i(\boldsymbol{x}_1,\ldots,\boldsymbol{x}_m) = \boldsymbol{x}_i$
and define the injection
$\mathrm{in}_i:\Cx^n \rightarrow \Cx^{m(n+1)}$ by
\[ \mathrm{in}_i(\boldsymbol{x}) = 
(\boldsymbol{0},\ldots,(\boldsymbol{x},1),\ldots,\boldsymbol{0}) \in \Cx^{m(n+1)}\, ,\]
where $(\boldsymbol{x},1)$ occurs in the $i$-th block.  We denote
$\mathrm{in}_1(\boldsymbol{0})$ by $\boldsymbol{v}_{\mathrm{init}}$.

\begin{proposition}
  Let $\mathcal{M}$ be the semigroup generated
  by the set of matrices $\{ M^{(e)}:e\in E\}$.  Then for
  $i=1,\ldots,m$ we have 
\[ S_{q_i} = \left\{\boldsymbol{x} \in \Q^n : \mathrm{in}_i(\boldsymbol{x})
\in \{M\boldsymbol{v}_{\mathrm{init}}:M\in\mathcal{M}\} \right\} \, .\]
\label{prop:comp-inv}
\end{proposition}
\begin{proof}
For an edge $e=(q_i,f,q_j)$ of the affine program $\mathcal{A}$ we have
\[ M^{(e)} \mathrm{in}_i(\boldsymbol{x}) =
           \mathrm{in}_j(f(\boldsymbol{x})) \]
and
\[ M^{(e)} \mathrm{in}_k(\boldsymbol{x}) = \boldsymbol{0}  \]
for $k\neq i$.
Now consider a sequence of edges
  \[e_1=(q_{i_1},f_1,q_{j_1}),(q_{i_2},f_2,q_{j_2}),\ldots,
    e_\ell=(q_{i_\ell},f_\ell,q_{j_\ell}) \, . \]
If this sequence is a
  legitimate execution of $\mathcal{A}$, i.e.,
  ${i_1}=1$ and $j_k=i_{k+1}$ for
  $k=1,\ldots,\ell-1$, then we have
\[ M^{(e_\ell)} \cdots M^{(e_1)}\boldsymbol{v}_{\mathrm{init}}
  =\mathrm{in}_{j_\ell}(f_\ell(\ldots f_1(\boldsymbol{0})\ldots)) \, . \]
If the sequence is not a legitimate execution of $\mathcal{A}$ then we have
\[ M^{(e_\ell)} \cdots M^{(e_1)}\boldsymbol{v}_{\mathrm{init}} = \boldsymbol{0} \, .\]

From the above it follows that for all $i\in\{1,\ldots,m\}$,
\[
  S_{q_i} = \left\{\boldsymbol{x} \in \Q^n : \mathrm{in}_i(\boldsymbol{x})
              \in \{ M\boldsymbol{v}_{\mathrm{init}} : M \in {\mathcal{M}} \} \right\} \, . \]
\end{proof}

\begin{theorem}
  Given an affine program $\mathcal{A}$ we can compute
  $\{ V_q : q\in Q\}$---the real Zariski closure of the collecting
  semantics.
\end{theorem}
\begin{proof}
From Proposition~\ref{prop:comp-inv} we have 
\begin{eqnarray*}
  S_{q_i} &=& \left\{\boldsymbol{x} \in \Q^n : \mathrm{in}_i(\boldsymbol{x})
              \in \{ M\boldsymbol{v}_{\mathrm{init}} : M \in {\mathcal{M}} \} \right\} \\
&=& \left\{ \boldsymbol{x} \in \Q^n : (\boldsymbol{x},1) \in 
   \Pi_i\left(\{M\boldsymbol{v}_{\mathrm{init}}:M\in\mathcal{M}\}\right) \right\} \, .
\end{eqnarray*}
By Theorem~\ref{thm:main} we can compute the complex Zariski closure
$\Zcl{\mathcal{M}}$ of the matrix semigroup $\mathcal{M}$.
Since the projection $\Pi_i$
and the map $M\mapsto M\boldsymbol{v}_{\mathrm{init}}$ are both
Zariski continuous, we have that 
\begin{eqnarray*}
S_{q_i} & \subseteq & 
\left\{ \boldsymbol{x} \in \Cx^n : (\boldsymbol{x},1) \in 
   \Pi_i\left(\{M\boldsymbol{v}_{\mathrm{init}}:M\in\Zcl{\mathcal{M}}\}\right) \right\} \\
&\subseteq & \Zcl{S_{q_i}} \, .
\end{eqnarray*}
Thus we can compute $\Zcl{S_{q_i}}$ as the complex Zariski closure of
\[ \left\{ \boldsymbol{x} \in \Cx^n : (\boldsymbol{x},1) \in 
   \Pi_i\left(\{M\boldsymbol{v}_{\mathrm{init}}:M\in\Zcl{\mathcal{M}}\}\right)
 \right\} \, ,\]
since the latter is a constructible set.

Finally we can compute $V_{q_i}=\Zcl[\R]{S_{q_i}}$---the real Zariski closure of
$S_{q_i}$---by intersecting the complex Zariski closure with $\R^n$.
As noted in Sec.~\ref{algorithmic_manipulation}, the intersection with
$\R^n$ is effective.

\end{proof}

\section{Zariski Closure of a Subgroup of $\GL_n(\Cx)$}
In this section we show how to compute the Zariski closure of the
subgroup of $\GL_n(\Cx)$ generated by a given constructible subset of
$\GL_n(\Cx)$.  We show this by reduction to the problem of computing
the Zariski closure of a finitely generated subgroup of
$\GL_n(\Cx)$.  An algorithm for the latter problem was given by
Derksen, Jeandel, and Koiran~\cite{DerksenJK05}.

Recall that for $X\subseteq\GL_n(\Cx)$ we use $\Gen{X}$ to denote the
\emph{sub-semigroup} of $\GL_n(\Cx)$ generated by $X$.  But we have:
\begin{lemma}[\cite{DerksenJK05}]\label{lem:closed_subsemigroup_is_subgroup}
A closed subsemigroup of $\GL_n(\Cx)$ is a subgroup.
\label{lem:closed-semi}
\end{lemma}
In particular, if $X\subseteq\GL_n(\Cx)$ then $\Zcl{\Gen{X}}$ is a
subgroup of $\GL_n(\Cx)$.


Our aim is to generalise the following result.
\begin{theorem}[\cite{DerksenJK05}]
\label{th:derksen}
Given matrices $a_1,\ldots,a_k\in\GL_n(\Cx)$ with algebraic entries,
we can compute the closed subgroup $\Zcl{\Gen{a_1,\ldots,a_k}}$.
\end{theorem}

The first generalisation is as follows.
\begin{corollary}
  Let $a_1,\ldots,a_k\in\GL_n(\Cx)$ and let $Y\subseteq \GL_n(\Cx)$ be an
  irreducible variety containing the identity $I_n$.  Assume that
  $a_1,\ldots,a_k$ have algebraic entries and that $Y$ is presented as
  the zero set of a finite collection of polynomials with algebraic
  coefficients.  We then have that $\Zcl{\Gen{a_1,\ldots,a_k,Y}}$ is computable
  from $Y$ and the $a_i$.
\label{cor:derksen}
\end{corollary}
\begin{proof}
  Let $G=\Zcl{\Gen{a_1,\ldots,a_k}}$ and let $H$ be the smallest
  Zariski closed subgroup of $\GL_n(\Cx)$ that contains $Y$ and is
  closed under conjugation by $a_1,\ldots,a_k$ (i.e., such that
  $a_iHa_i^{-1}\subseteq H$ for $i=1,\ldots,k$).  We claim that
  $\Zcl{\Gen{a_1,\ldots,a_k,Y}} = \Zcl{G\cdot H}$.

  To prove the claim, note that since $H$ is closed under conjugation
  by $a_1,\ldots,a_k$ then $H$ is also closed under conjugation by any
  $g \in \Gen{a_1,\ldots,a_k}$.  Moreover, since the map
  $g\mapsto ghg^{-1}$ is Zariski continuous for each fixed $h\in H$,
  we have that $H$ is closed under conjugation by any
  $g\in G=\Zcl{\Gen{a_1,\ldots,a_k}}$.  It follows that $G\cdot H$ is
  a sub-semigroup of $\GL_n(\Cx)$ and so $\Zcl{G\cdot H}$ is a group
  by Lemma~\ref{lem:closed-semi}.  But
\[\{a_1,\ldots,a_k\} \cup Y \subseteq G\cdot H \subseteq
\Zcl{\Gen{a_1,\ldots,a_k,Y}} \] and hence
$\Zcl{G\cdot H}=\Zcl{\Gen{a_1,\ldots,a_k,Y}}$.

It remains to show that we can compute $\Zcl{G\cdot H}$.  Now we can
compute $G$ by Theorem~\ref{th:derksen}. To compute $H$ we use the
following algorithm:

\RestyleAlgo{ruled}
\begin{procedure}
\Input{Irreducible variety $Y\subseteq\GL_n(\Cx)$ containing $I_n$}
\Input{$a_1,\ldots,a_k\in\GL_n(\Cx)$}
\label{alg:reduce:initH} $H:=Y$\; $S=\{a_1,\ldots,a_k,I_n\}$\;
\label{alg:reduce:loop:start} \Repeat{$H_{old}=H$}{
    $H_{old}:=H$\;
    \For{$y\in S$}{
        $H:=\Zcl{H\cdot yHy^{-1}}$\label{alg:reduce:newH2}\;
    }
}
\label{alg:reduce:loop:end}
\Output{$H$}
\caption{FinPlusIrredClosure($a_1,\ldots,a_k,Y$)}
\label{alg:reduce}
\end{procedure}

We show that Algorithm~\ref{alg:reduce} computes the\linebreak smallest subgroup
$H$ of $\GL_n(\Cx)$ that is Zariski closed, contains $Y$, and is
closed under conjugation by $a_1,\ldots,a_k$.  To this end, notice
that since $Y$ contains the identity the successive values taken by
$H$ in the algorithm form an increasing chain of sub-varieties of
$\GL_n(\Cx)$. Moreover by Lemma~\ref{lem:prod_irred_irred} this chain
is in fact an increasing chain of \emph{irreducible} sub-varieties.
But such a chain has bounded length since $\GL_n(\Cx)$ has finite
dimension and hence the algorithm must terminate.

We know that $Y\subseteq H$ on termination.  Moreover, from the loop
termination condition, it clear that on termination $H$ must be\linebreak closed
under conjugation by $a_1,\ldots,a_k$, and be a Zariski closed
sub-semigroup of $\GL_n(\Cx)$ (and hence a sub-group of $\GL_n(\Cx)$
by Lemma~\ref{lem:closed_subsemigroup_is_subgroup}).  Finally, by construction, $H$ is the
smallest such subgroup of $\GL_n(\Cx)$.  This concludes the proof.
\end{proof}

We can now prove the main result of this section.
\begin{theorem}\label{th:closure_definable_subgroup}
Given a constructible subset $A$ of $\GL_n(\Cx)$, we can 
compute $\Zcl{\Gen{A}}$.
\end{theorem}
\begin{proof}
  Let $X_1,\ldots,X_k$ be the irreducible components of $\Zcl{A}$,
  which are computable from $A$.  For each $i$, compute a point
  $a_i\in X_i$ with algebraic entries (using, e.g., the procedure
  of~\cite[Chapter 12.6]{BPR06}).  Form $Y_i=a_i^{-1}X_i$ which is an irreducible
  variety containing the identity and let
  $Y=\Zcl{Y_1\cdot Y_2\cdots Y_k}$ which by
  Lemma~\ref{lem:prod_irred_irred} is also an irreducible variety
  containing the identity. We then have that
  $\Zcl{\Gen{A}}=\Zcl{\Gen{a_1,\ldots,a_k,Y}}$.  Indeed, clearly
 $\Gen{A}=\Gen{a_1,\ldots,a_k,Y_1\cdot Y_2\cdots Y_k}$ and thus
\begin{align*}
\Zcl{\Gen{A}}
&=\Zcl{\Gen{a_1,\ldots,a_k,Y_1\cdot Y_2\cdots Y_k}}\\
&=\Zcl{\Gen{a_1,\ldots,a_k,\Zcl{Y_1\cdot Y_2\cdots Y_k}}}.
\end{align*}
We can compute the closure of $\Gen{a_1,\ldots,a_k,Y}$ thanks to
Corollary~\ref{cor:derksen}.
\end{proof}

\section{Zariski Closure of a Finitely Generated Matrix Semigroup}
\label{sec:main-algorithm}
In this section we give a procedure to compute the Zariski closure of
a finitely generated matrix semigroup.  We proceed by induction on the
rank of the generators. To this end, it is useful to generalise from
finite sets of generators to constructible sets of generators.  In
particular, we will use Theorem~\ref{th:closure_definable_subgroup} on
the computability of the Zariski closure of the group generated by a
constructible set of matrices.  

We first introduce a graph structure on the set of generators that
allows us to reason about all products of generators that have a given
rank.
\subsection{A Generating Graph}
Given integers $n$ and $r$, let $A\subseteq M_n(\Cx)$ be a set of
matrices of rank $r$.  We define a labelled directed graph
$\mathcal{K}(A)$ as follows:
\begin{itemize}
\item There is a vertex $(U,V)$ for each pair of subspaces
  $U,V\subseteq \Cx^n$ such that $\dim(V)=r$, $\dim(U)=n-r$, and $U\cap
  V=0$.
\item There is a labelled edge $(U,V) \xrightarrow{a} (U',V')$
for each pair of vertices $(U,V)$ and $(U',V')$, and  
 each matrix $a\in A$ such that $\ker(a)=U$ and
  $\im(a)=V'$.
\end{itemize}
We note in passing that $\mathcal{K}(A)$ can be seen as an
edge-induced subgraph of the \emph{Karoubi
  envelope}~\cite{Steinberg16} of the semigroup $M_n(\Cx)$.

A \emph{path} in $\mathcal{K}(A)$ is a non-empty sequence of
consecutive edges
\[(U_0,V_0) \xrightarrow{a_1} (U_1,V_1) \xrightarrow{a_2} (U_2,V_2)
\xrightarrow{a_3} \ldots \xrightarrow{a_m} (U_m,V_m).\] The length of
such a path is $m$ and its \emph{label} is the product $a:=a_m \cdots
a_1$.  Matrix $a$ has rank $r$ since $\ker(a_{i+1}) \cap \im(a_i)=0$
for $i=1,\ldots,m-1$.  It is moreover clear that $\{a\in\Gen{A} :
\rank(a)=r\}$ is precisely the set of labels over all paths in
$\mathcal{K}(A)$.  We will denote that there is a path from $(U,V)$
to $(U',V')$ with label $a$ by writing $(U,V) \xRightarrow{a}
(U',V')$.

  The following sequence of propositions concerns the structure of the
  SCCs in $\mathcal{K}(A)$.  The respective proofs make repeated
  use of the fact that for each vertex $(U,V)$ of $\mathcal{K}(A)$
  we have $\iota(U) \wedge \iota(V) \neq 0$ and that $\dim
  \ExtAlg[r](\Cx^n)=\binom{n}{r}$ (cf.~Sec.~\ref{sec:misc}).  We say
  that an SCC of $\mathcal{K}(A)$ is \emph{non-trivial} if it contains
  a vertex $(U,V)$ such that there is a path from $(U,V)$ back to
  itself. Figure~\ref{fig:semigroup} summarises the structural results on $\mathcal{K}(A)$.

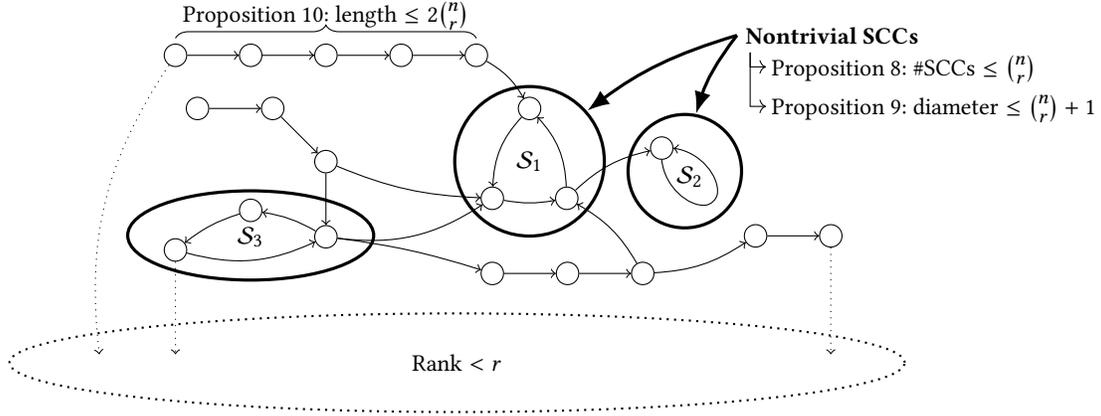
\begin{figure*}
\begin{tikzpicture}[
        kaedge/.style={-{Classical TikZ Rightarrow}},
        kapath/.style={-{Classical TikZ Rightarrow}},
        kasccedge/.style={very thick,-Latex},
        karankedge/.style={dotted, -{Classical TikZ Rightarrow}},
    ]
    \node (a0) [draw,circle] {};
    \node[right of=a0] (a1) [draw,circle] {}; \draw[kaedge] (a0) -- (a1);
    \node[right of=a1] (a2) [draw,circle] {}; \draw[kaedge] (a1) -- (a2);
    \node[right of=a2] (a3) [draw,circle] {}; \draw[kaedge] (a2) -- (a3);
    \node[right of=a3] (a4) [draw,circle] {}; \draw[kaedge] (a3) -- (a4);

    \node[below left of=a1] (d1) [draw,circle] {};
    \node[right of=d1] (d2) [draw,circle] {}; \draw[kaedge] (d1) -- (d2);
    \node[below right of=d2] (d3) [draw,circle] {}; \draw[kaedge] (d2) -- (d3);

    \node[below right of=a4] (c1a) [draw,circle] {};
    \node[node distance=0.7cm, below of=c1a] (c1) {$\mathcal{S}_1$};
    \node[node distance=0.7cm, below left of=c1] (c1b) [draw,circle] {}; \draw[kapath] (c1a) to[bend right=20] (c1b);
    \node[node distance=0.7cm, below right of=c1] (c1c) [draw,circle] {}; \draw[kapath] (c1b) to[bend right=10] (c1c);
    \draw[kapath] (c1c) to[bend right=20] (c1a);
    \node[every fit/.style={circle,draw,very thick,inner sep=0pt,yshift=-3pt}, fit=(c1a) (c1b) (c1c)] (c1circ) {};
    \draw[kaedge] (a4) to[bend left=15] (c1a);
    \draw[kaedge] (d3) to[bend right=10] (c1b);

    \node[below of=d3] (c3a) [draw,circle] {};
    \node[left of=c3a] (c3) {$\mathcal{S}_3$};
    \node[above of=c3,node distance=10pt] (c3c) [draw,circle] {}; \draw[kapath] (c3a) to[bend right=10] (c3c);
    \node[left of=c3,yshift=-5pt] (c3b) [draw,circle] {}; \draw[kapath] (c3c) to[bend right=10] (c3b);
    \draw[kapath] (c3b) to[bend right=20] (c3a);
    \node[every fit/.style={ellipse,draw,very thick,inner sep=0pt,yshift=-2pt}, fit=(c3a) (c3b) (c3c)] (c3circ) {};
    \draw[kaedge] (d3) -- (c3a);
    \draw[kaedge] (c3a) to[bend right=20] (c1b);

    \node[below of=c1b] (e1) [draw,circle] {}; \draw[kaedge] (c3a) to[bend left=5] (e1);
    \node[right of=e1] (e2) [draw,circle] {}; \draw[kaedge] (e1) -- (e2);
    \node[right of=e2] (e3) [draw,circle] {}; \draw[kaedge] (e2) -- (e3);
    \draw[kaedge] (e3) to[bend right=15] (c1c);

    \node[right of=e3, node distance=1.5cm, yshift=0.5cm] (b1) [draw,circle] {}; \draw[kaedge] (e3) to[bend right=15] (b1);
    \node[right of=b1] (b2) [draw,circle] {}; \draw[kaedge] (b1) -- (b2);

    \node[node distance=50pt,right of=c1a,yshift=-15pt] (c2a) [draw,circle] {};
    \node[below right of=c2a,node distance=15pt] (c2) {$\mathcal{S}_2$};
    \node[below right of=c2,node distance=15pt,inner sep=0pt] (c2bhidden) {};
    \draw[kaedge,rounded corners] (c2a) to[bend right=45] (c2bhidden) to[bend right=45] (c2a);
    \node[every fit/.style={ellipse,draw,very thick,inner sep=2pt,yshift=0pt}, fit=(c2a) (c2bhidden)] (c2circ) {};

    \draw[kaedge] (c1c) to[bend left=15] (c2a);

    \draw[decorate,decoration={brace,amplitude=5pt}] ([yshift=2pt]a0.north) -- ([yshift=2pt]a4.north)
        node[above,midway]
        {Proposition~\ref{prop:long_path_cross_scc}:
          $\text{length}\leq 2{n\choose r}$};

    \node[above of=c2a, xshift=1cm, node distance=1.5cm,anchor=west] (scctext) {\textbf{Nontrivial SCCs}};
    \node[anchor=north west,xshift=10pt,yshift=2pt] (sccprop) at (scctext.south west) {
        Proposition~\ref{prop:bound_number_nontrivial_sccs}:
        $\#\text{SCCs}\leq \binom{n}{r}$};
    \draw[->] ([xshift=5pt]scctext.south west) |- ([yshift=1pt,xshift=1pt]sccprop.west);
    \node[anchor=north west,yshift=2pt] (sccprop2) at (sccprop.south west) {
        Proposition~\ref{prop:short_path_in_scc}: $\text{diameter}\leq
        \binom{n}{r}+1$};
    \draw[->] ([xshift=5pt]scctext.south west) |- ([yshift=1pt,xshift=1pt]sccprop2.west);

    \draw[kasccedge] (scctext.west) to[bend right=10] (c1circ);
    \draw[kasccedge] (scctext.west) to[bend right=10] (c2circ);

    \node[below of=c3b, node distance=1.5cm] (rankb) {}; \draw[karankedge] (c3b) -- (rankb);
    \node[left of=rankb] (ranka) {}; \draw[karankedge] (a0) to[bend right=20] (ranka);
    \node (rankc) at (rankb -| b2) {}; \draw[karankedge] (b2) -- (rankc);

    \node[every fit/.style={ellipse,draw,thick, dotted,
        inner sep=10pt,yshift=0pt}, fit={([xshift=1cm]ranka.east) (rankb) ([xshift=-1cm]rankc.west)}] (rank) {};
    \draw (rank.center) node (ranktext) {$\text{Rank}<r$};
\end{tikzpicture}
\caption{\label{fig:semigroup}Graphical representation of $\mathcal{K}(A)$,
vertex and edge labels omitted for clarity. Note that the graph can have infinitely many vertices.
Propositions~\ref{prop:bound_number_nontrivial_sccs}~and~\ref{prop:short_path_in_scc} respectively
show that are only finitely many nontrivial SCCs and they have finite diameter.
Proposition~\ref{prop:long_path_cross_scc} shows that paths avoiding nontrivial SCCs must be short.
All paths in $\mathcal{K}(A)$ are labelled by rank $r$ matrices. Dotted arrows represent products
in the semigroup where the rank becomes less than $r$: those products do not correspond to labels in
$\mathcal{K}(A)$ and need to be handled separately.}
\end{figure*}

\begin{proposition}\label{prop:bound_number_nontrivial_sccs}
  $\mathcal{K}(A)$ has at most $\binom{n}{r}$ non-trivial SCCs.  
\end{proposition}
\begin{proof}
  Let $(U_1,V_1),\ldots,(U_m,V_m)$ be an arbitrary finite set of
  vertices drawn from distinct non-trivial SCCs of $\mathcal{K}(A)$.
  To prove the proposition it suffices to show that
  $m\leq \binom{n}{r}$.

Assume that the
  vertices $(U_1,V_1),\ldots,(U_m,V_m)$ are given according to a
  topological ordering of SCCs---so that there is no path from
  $(U_j,V_j)$ back to $(U_i,V_i)$ for $i<j$.
  By assumption, for $i=1,\ldots,m$ there exists a path
  $(U_i,V_i)\xRightarrow{a_i} (U_i,V_i)$.  

On the one hand, for all $1 \leq i < j \leq m$, we have $\iota(U_i) \wedge \iota(V_j) =0$ (equivalently,
$U_i\cap V_j \neq 0$)---for otherwise 
there would be path \[(U_j,V_j)
\xRightarrow{a_j} (U_i,V_j) \xRightarrow{a_i} (U_i,V_i) \, , \]
contrary to the topological ordering.  On the other hand we have
that $\iota(U_j) \wedge \iota(V_j) \neq 0$ (equivalently, $U_j \cap V_j =
0$) for all $j\in\set{1,\ldots,m}$ by definition of $\mathcal{K}(A)$.
It follows that \[\iota(U_j) \not\in \Span\{\iota(U_i) :
i=1,\ldots,j-1\}\] for all $j\in\set{1,\ldots,m}$. Indeed, by the claim, any element $U$ in this span
satisfies $\iota(U)\wedge\iota(V_j)=0$ by bilinearity of the wedge
product.  We conclude that \[ \dim \Span \{ \iota(U_i) \in
\ExtAlg[r](\Cx^n) : i = 1,\ldots,j\} = j \] for all $1 \leq j \leq m$
and hence $m\leq \dim \ExtAlg[r](\Cx^n)=
\binom{n}{r}$, as we wished to prove.
\end{proof}

\begin{proposition}
\label{prop:short_path_in_scc}
If there is a path from $(U,V)$ and $(U',V')$ in $\mathcal{K}(A)$,
then there is a path from $(U,V)$ to $(U',V')$ of length at most
$\binom{n}{r}+1$.
\end{proposition}
\begin{proof}
  Let 
\begin{gather}
(U_0,V_0) \xrightarrow{a_1} (U_1,V_1) \xrightarrow{a_2}
  \ldots \xrightarrow{a_m}(U_m,V_m)
\label{eq:short}
\end{gather}
 be a shortest path from
  $(U_0,V_0)=(U,V)$ to $(U_m,V_m)=(U',V')$.  By construction we have
  that $U_i \cap V_i=0$ for $i=0,\ldots,m$.  Furthermore we have
  $U_j\cap V_i\neq 0$ for all $0< i<j <m$,
  for otherwise we would have a shortcut
\[ (U_{i-1},V_{i-1}) \xrightarrow{a_i} (U_j,V_i) \, 
    \xrightarrow{a_{j+1}} (U_{j+1},V_{j+1}) \, , \] contradicting the
    minimality of (\ref{eq:short}).  But then $\iota(V_j) \not\in
    \Span\{\iota(V_i) : 1\leq i<j\}$ for $j=1,\ldots,m-1$: indeed any
    element $V$ in this span satisfies $\iota(U_j)\wedge \iota(V)=0$
    by bilinearity of the wedge product, but we know that
    $\iota(U_j)\wedge\iota(V_j)\neq 0$. We conclude that
  \[\dim \Span \{ \iota(V_i) \in \ExtAlg[r](\Cx^n) : i\in\{1,\ldots,j\}\}=j\]
  for all $j=1,\ldots,m-1$. It follows that $m-1\leq {n\choose r}$.
\end{proof}

\begin{proposition}
\label{prop:long_path_cross_scc}
  Given any path $(U_0,V_0)\xrightarrow{a_1} (U_1,V_1)
  \xrightarrow{a_2} \ldots \xrightarrow{a_m} (U_m,V_m)$ in
  $\mathcal{K}(A)$, where $m=2\binom{n}{r}$, some vertex $(U_i,V_i)$
  lies in a non-trivial SCC.
\end{proposition}
\begin{proof}
The set of $\binom{n}{r}+1$ vectors
$\set{\iota(U_0),\iota(U_2),\iota(U_4),\ldots,\iota(U_m)}$ is linearly dependent
since $\dim \ExtAlg[r](\Cx^n)=\binom{n}{r}$.  Thus there must exist
$i\in \set{0,\ldots,m}$ such that $\iota(U_i) \in\Span\set{\iota(U_j):
  j \leq i-2}$.  Now by definition of $\mathcal{K}(A)$ we have
$U_i\cap V_i=0$ and hence $\iota(U_i) \wedge
\iota(V_i)\neq 0$.  Thus by bilinearity of the wedge product there must exist $j\leq i-2$
such that $\iota(U_j) \wedge \iota(V_i)\neq 0$, that is, $U_j\cap
U_i=0$.  But then we have a path
\[ (U_{i-1},V_{i-1}) \xrightarrow{a_i} (U_j,V_i) \xrightarrow{a_{j+1}}
  (U_{j+1},V_{j+1}) \, ,\] showing that $(U_{i-1},V_{i-1})$ and
$(U_{j+1},V_{j+1})$ lie in the same (necessarily non-trivial)
SCC. Indeed, recall that $j\leq i-2$ so either
$(U_{j+1},V_{j+1})\xRightarrow{}(U_{i-1},V_{i-1})$ or
$(U_{j+1},V_{j+1})=(U_{i-1},V_{i-1})$ in the original path.
\end{proof}

\subsection{Adding Pseudo-Inverses}
We now focus on individual SCCs within $\mathcal{K}(A)$.  Let
$\mathcal{S}$ be such a non-trivial SCC.  For each edge
$(U,V) \xrightarrow{a} (U',V')$ in $\mathcal{S}$, define its
\emph{pseudo-inverse} to be a directed edge
$(U',V') \xrightarrow{a^{+}} (U,V)$, where $a^+\in M_n(\Cx)$ is the
unique matrix such that $\ker(a^+)=U'$, $\im(a^+)=V$, $a^+av=v$ for
all $v\in V$, and $aa^+v = v$ for all $v\in V'$.  We write
${\mathcal{S}^+}$ for the graph obtained from $\mathcal{S}$ by adding
pseudo-inverses of every edge in $\mathcal{S}$.

The graph ${\mathcal{S}^+}$ can be seen as the generator of a groupoid
in which the above-defined pseudo-inverse matrices are genuine
 inverses.  We do not develop this idea, except to observe that
not only edges but also paths in $\mathcal{S}$ have pseudo-inverses in
${\mathcal{S}^+}$.  Specifically, given a path
$(U,V) \xRightarrow{a} (U',V')$ in $\mathcal{S}$, one obtains a path
$(U',V') \xRightarrow{a^+} (U,V)$ in ${\mathcal{S}^+}$ by taking the
pseudo-inverse of each constituent edge.  In the remainder of this
section we show that the pseudo-inverses of all paths in $\mathcal{S}$
are already present in the Zariski closure $\Zcl{\Gen{A}}$.

\begin{proposition}\label{prop:two-closures-are-us}
  Let $(U,V)$ be a vertex of $\mathcal{S}$ and let
  $B\subseteq M_n(\Cx)$ be a constructible set of matrices such that
  there is a path $(U,V) \xRightarrow{b} (U,V)$ in $\mathcal{S}$ for
  all $b \in B$.  Then $\Zcl{\Gen{B}}$ is computable from $B$ and for
  every $b\in \Gen{B}$ the pseudo-inverse
  $(U,V) \xRightarrow{b^+} (U,V)$ is such that
  $b^+ \in \Zcl{\Gen{B}}$.
\end{proposition}
\begin{proof}
  By construction, all elements of $B$ have kernel $U$ and image $V$,
  where $U\oplus V=\Cx^n$.  Thus there is an invertible matrix
  $y \in \GL_n(\Cx)$ such that for every $b\in B$ there exists
  $c \in \GL_r(\Cx)$ with
\[ y^{-1}b y = \begin{bmatrix} c & 0 \\ 
                             0 & 0 \end{bmatrix} \, . \]
Let 
\[ C:= \left\{ c\in\GL_r(\Cx) : \exists b\in B\, . y^{-1}b y = \begin{bmatrix} c & 0 \\
      0 & 0 \end{bmatrix} \right \} \, , \]
which is constructible. We can compute $\Zcl{\Gen{C}}$
(the Zariski closure of $\Gen{C}$ in the variety $\GL_r(\Cx)$) using
Theorem~\ref{th:closure_definable_subgroup}.  But then
\[ \left\{ y \begin{bmatrix} c & 0 \\ 0 & 0 \end{bmatrix} y^{-1} : c
    \in \Zcl{\Gen{C}} \right\} \] is a constructible subset of
$M_n(\Cx)$ whose closure equals $\Zcl{\Gen{B}}$. Note that we are
using the fact that $\Zcl{\Gen{C}}$ is a subvariety of $\GL_n(\Cx)$
thus it is constructible in $M_n(\Cx)$.  Finally, if
$b=y\begin{bmatrix} c & 0 \\ 0 & 0 \end{bmatrix} y^{-1}\in \Gen{B}$
then
$b^+=y\begin{bmatrix} c^{-1} & 0 \\ 0 & 0 \end{bmatrix}
y^{-1}\in\Zcl{\Gen{B}}$ since $c^{-1}\in\Zcl{\Gen{C}}$ (which is a
group by Lemma~\ref{lem:closed_subsemigroup_is_subgroup}).
\end{proof}

\begin{corollary}
  Suppose that $(U,V) \xRightarrow{a} (U',V')$ is a path in
  $\mathcal{S}$ with pseudo-inverse $(U',V') \xRightarrow{a^+} (U,V)$.
  Then $a^+ \in \Zcl{\Gen{A}}$.
\label{corl:include}
\end{corollary}
\begin{proof}
Since $\mathcal{S}$ is strongly connected, there is a path $(U',V')
\xRightarrow{b} (U,V)$.  Consider the path $(U,V) \xRightarrow{ba}
(U,V)$ and its pseudo-inverse $(U,V) \xRightarrow{(ba)^+} (U,V)$.  By
Proposition~\ref{prop:two-closures-are-us} we have $(ba)^+ \in
\Zcl{\Gen{A}}$.  We moreover have $a^+ = a^+b^+b=(ba)^+b$ and hence
$a^+ \in \Zcl{\Gen{A}}$, since $\Zcl{\Gen{A}}$ is a semigroup.
\end{proof}

\subsection{Maximum-Rank Matrices in the Closure}
Let $\mathcal{S}$ be a non-trivial SCC in $\mathcal{K}(A)$.  Write
$B \subseteq M_n(\Cx)$ for the set of labels of all paths in
${\mathcal{S}^+}$ of length at most $\binom{n}{r}+2$.  Moreover fix a
vertex $(U_*,V_*)$ in ${\mathcal{S}^+}$ and write $B_*$ for the set of
labels of all paths in ${\mathcal{S}^+}$ of length at most
$2{n\choose r}+3$ that are self-loops on $(U_*,V_*)$.

\begin{proposition}
\label{prop:closure_of_scc}
Let $\Gen{\mathcal{S}}$ denote the set of labels of all paths in $\mathcal{S}$.
Then 
\[ \Gen{\mathcal{S}} \subseteq B\Zcl{\Gen{B_*}} 
B \subseteq \Zcl{\Gen{A}}\]
\end{proposition}
\begin{proof}
  By Corollary~\ref{corl:include} we have that
  $B,B_* \subseteq \Zcl{\Gen{A}}$.  Thus the right-hand inclusion
  follows from the fact that $\Zcl{\Gen{A}}$ is a semigroup.

To establish the left-hand inclusion,
consider a path
\[ (U_0,V_0) \xrightarrow{a_1} (U_1,V_1) \xrightarrow{a_2} (U_2,V_2)
  \xrightarrow{a_3} \ldots \xrightarrow{a_n} (U_n,V_n) \] within
$\mathcal{S}$.  Proposition~\ref{prop:short_path_in_scc} ensures that
for each vertex $(U_i,V_i)$ there is a path
$(U_*,V_*) \xRightarrow{f_i} (U_i,V_i)$ in $\mathcal{S}$ of length at
most ${n\choose r}+1$.  Such a path has a pseudo-inverse
$(U_i,V_i) \xRightarrow{f_i^+} (U_*,V_*)$ in ${\mathcal{S}^+}$.  Now by
the definition of a pseudo-inverse we have $a_if_{i-1}f_{i-1}^+=a_i$
for all $i\in \set{1,\ldots,n}$.  Thus
\begin{eqnarray*}
a_n \ldots a_2a_1 &=& a_n f_{n-1}f_{n-1}^+
                      a_{n-1} f_{n-2} f_{n-2}^+ 
                      \cdots 
                      f_2 f_2^+ a_2 f_1f_1^+ a_1 \\
          &=&        a_nf_{n-1}(f_{n-1}^+ a_{n-1} f_{n-2}) \cdots 
                      (f_2^+ a_2 f_1) f_1^+ a_1 \, .
\end{eqnarray*}
The result follows from the observation that $a_nf_{n-1}$ and
$f_1^+a_1$ are both elements of $B$ and that
$f_i^+ a_i f_{i-1} \in B_*$ for $i=2,\ldots,n-1$.
\end{proof}

Recall from Proposition~\ref{prop:short_path_in_scc} that the graph
$\mathcal{K}(A)$ has at most $\binom{n}{r}$ non-trivial SCCs.
Let $\mathcal{S}_1,\ldots,\mathcal{S}_\ell$ be a list of the
non-trivial SCCs in $\mathcal{K}(A)$ and write
\begin{gather}
\PP := A \cup \Gen{\mathcal{S}_1} \cup \cdots \cup \Gen{\mathcal{S}_\ell} \, .
\label{eq:all-paths}
\end{gather}

\begin{lemma}\label{lem:structure_paths}
Given $a\in\Gen{A}$ with $\rank(a)=r$, we have $a \in \PP \cup \PP^2
\cup \cdots \cup \PP^\kappa$, where $\kappa=2\binom{n}{r}^2$.
\end{lemma}
\begin{proof}
Suppose that $a$ is the label of a path
\begin{gather}
 (U_0,V_0) \xrightarrow{a_1} (U_1,V_1) \xrightarrow{a_2} (U_2,V_2)
  \xrightarrow{a_3} \ldots \xrightarrow{a_m} (U_m,V_m)
\end{gather}
 in $\mathcal{K}(A)$.  The vertices along this path can be partitioned
 into maximal blocks of contiguous vertices all lying in the same SCC
 of $\mathcal{K}(A)$.  By Proposition~\ref{prop:long_path_cross_scc}
 there are at most $\binom{n}{r}$ such blocks corresponding to
 non-trivial SCCs.  The remaining blocks, corresponding to trivial
 SCCs, are singletons.  By Proposition~\ref{prop:long_path_cross_scc}
 there can be at most $2\binom{n}{r}$ consecutive blocks corresponding to
 trivial SCCs.  Thus there are at most $\kappa=2\binom{n}{r}^2$ blocks
 in total.

Now we can factor the path into single edges that connect vertices in
different blocks and sub-paths all of whose vertices lie in the same
block.  There are at most $\kappa$ such factors (the same as
the number of blocks) and the label of each factor lies in the set
$\PP$ defined in \eqref{eq:all-paths}.  This completes the proof.
\end{proof}

Let $R_r=\set{x\in M_n(\Cx):\rank(x)=r}$ which is a constructible set, and
$R_{<r}=\set{x\in M_n(\Cx):\rank(x)<r}$ which is closed.

\begin{proposition}
Let $A\subseteq M_n(\Cx)$ be a constructible set of matrices, all of rank $r$.
Then we can compute $\Zcl{\Gen{A}}\cap R_r$ from $A$.
\label{prop:compute-closure-fixed-rank}
\end{proposition}
\begin{proof}
By Proposition~\ref{prop:closure_of_scc}, for $i=1,\ldots,\ell$ we can
compute a constructible set $E_i \subseteq M_n(\Cx)$ such that
$\Gen{\mathcal{S}_i} \subseteq E_i \subseteq \Zcl{\Gen{A}}$.  Writing
$E:= A \cup E_1 \cup \ldots \cup E_\ell$, we have $\PP \subseteq E
\subseteq \Zcl{\Gen{A}}$.

By Lemma~\ref{lem:structure_paths} we have $\Gen{A}\cap R_r\subseteq X$, where $X := E
\cup E^2 \cup \ldots \cup E^{2\binom{n}{r}^2}$.  Now
\begin{align*}
\Gen{A}\cap R_r &\subseteq X\subseteq\Zcl{\Gen{A}}\\
\Zcl{\Gen{A}\cap R_r}
    &\subseteq \Zcl{X}\subseteq\Zcl{\Gen{A}}\\
\Zcl{\Gen{A}\cap R_r}\cap R_r
    &\subseteq \Zcl{X}\cap R_r\subseteq \Zcl{\Gen{A}}\cap R_r.
\end{align*}
We claim that
\begin{equation}\label{eq:compute-closure-fixed-rank:rank_closure}
\Zcl{\Gen{A}\cap R_r}\cap R_r=\Zcl{\Gen{A}}\cap R_r
\end{equation}
which shows that
\[\Zcl{\Gen{A}}\cap R_r=\Zcl{X}\cap R_r\]
is constructible and computable. It remains to see why \eqref{eq:compute-closure-fixed-rank:rank_closure}
holds. Since all matrices in $A$ have rank $r$, all matrices in $\Gen{A}$ have rank $r$ or less, thus
\begin{align*}
\Gen{A}&=\big(\Gen{A}\cap R_r\big)\cup\big(\Gen{A}\cap R_{<r}\big)\\
\Zcl{\Gen{A}}&=\Zcl{\Gen{A}\cap R_r}\cup\Zcl{\Gen{A}\cap R_{<r}}\\
\Zcl{\Gen{A}}\cap R_r&=\left(\Zcl{\Gen{A}\cap R_r}\cap R_r\right)\cup\underbrace{\left(\Zcl{\Gen{A}\cap R_{<r}}\cap R_r\right)}_{=\varnothing}.
\end{align*}
Indeed, $\Gen{A}\cap R_{<r}\subseteq R_{<r}$
thus $\Zcl{\Gen{A}\cap R_{<r}}\subseteq R_{<r}$ because $R_{<r}$ is closed, and $R_{<r}\cap R_r=\varnothing$.

\end{proof}

%

\subsection{Computing the Closure}
\label{mainsection}
We now present the main result of the paper.
\begin{theorem}
  Given a constructible set of matrices $A\subseteq M_n(\Cx)$, one can
  compute $\Zcl{\Gen{A}}$---the Zariski closure of the semigroup
  generated by $A$.
\label{thm:main}
\end{theorem}
\begin{proof}
  The proof is by induction on the maximum rank $r$ of the matrices in
  $A$.  The base case $r=0$ is trivial.  For the induction step, write
  $A_r:=\set{a\in A:\rank(a)=r}$ for the subset of matrices in $A$ of
  maximum rank and $B:=\{ a \in \Zcl{\Gen{A_r}} : \rank(a)=r\}$.  Now
  $B$ is computable by
  Proposition~\ref{prop:compute-closure-fixed-rank}.

  We claim that $\Zcl{\Gen{A}} = \Zcl{B} \cup \Zcl{\Gen{C}}$, where
  \[ C = \set{a \in A\cup BA \cup AB \cup BAB : \rank(a) < r} \, . \]
  The theorem follows from the claim since $\Zcl{\Gen{C}}$ is
  computable by the induction hypothesis.

  It remains to prove the claim.  For the right-to-left inclusion
  notice that since $A,B\subseteq \Zcl{\Gen{A}}$ and $\Zcl{\Gen{A}}$
  is a Zariski-closed semigroup, then $\Zcl{\Gen{A}}$ contains both
  $\Zcl{B}$ and $\Zcl{\Gen{C}}$.

  For the left-to-right inclusion it suffices to show that
  ${\Gen{A}} \subseteq \Zcl{B} \cup \Zcl{\Gen{C}}$.  To this end,
  consider a non-empty product $a:=a_1a_2 \cdots a_m$, where
  $a_1,\ldots,a_m \in A$.  Suppose first that $\rank(a)=r$. Then of
  course 
  $a_1,\ldots,a_m \in A_r$ and hence $a\in B$.  Suppose now that that
  $\rank(a)<r$.  We show that $a \in \Gen{C}$ by induction on $m$.
  Let $a_1 \cdots a_\ell$ be a prefix of minimum length that has rank
  less than $r$.  Clearly such a prefix lies in $A \cup BA$.  Moreover
  the corresponding suffix $a_{\ell+1} \cdots a_m$ is either empty,
  has rank $r$ (and hence is in $B$), or has rank $<r$ and
  hence is in $\Gen{C}$ by induction.  In all cases we have that
  $a \in \Gen{C}$.
\end{proof}

\begin{corollary}
  Given a constructible set of matrices $A\subseteq M_n(\R)$, one can
  compute $\Zcl[M_n(\R)]{\Gen{A}}$---the real Zariski closure of the semigroup
  generated by $A$.
\label{cor:main}
\end{corollary}
\begin{proof}
  For any set $X\subseteq\R^n$, we have $\Zcl[M_n(\R)]{X} = \Zcl[M_n(\Cx)]{X} \cap M_n(\R)$ (see
  Secs.~\ref{sub:alg-geo} and \ref{algorithmic_manipulation}).
\end{proof}

\section{Conclusion}
The main technical contribution of this paper is a procedure to compute the
Zariski closure of the semigroup generated by a given finite set of
rational square matrices of the same dimension.  We have not attempted
to analyse the complexity of this procedure.  Such an analysis would
depend on, among other things, the various Gr\"{o}bner basis
manipulations that we perform and the algorithm of~\cite{DerksenJK05}
for computing the Zariski closure of a finitely generated group of
invertible matrices, which we use as a subroutine.  The task of
computing Gr\"{o}bner bases is known to be expensive, while 
there has
been no complexity analysis of the algorithm of~\cite{DerksenJK05} to
the best of our knowledge.
It may be that the techniques developed in this paper can be used to
obtain computable bounds on the degree of the generators of an ideal
representing the Zariski closure of a given finitely generated matrix
semigroup.  If this were the case then one could compute a set of
generators essentially using only linear algebra (in the spirit of the
algorithm of~\cite{Muller-OlmS04} for computing polynomial invariants
of a given maximum degree for a given affine program).

\bibliography{literature}

\end{document}